\documentclass[11pt,a4paper,reqno]{amsart}
\usepackage{amsmath,amsthm,amsfonts,amssymb,bbm}
\usepackage{graphicx,psfrag,subfigure,color}
\usepackage{cite,stackrel}
\usepackage{hyperref}

\numberwithin{equation}{section}

\newtheorem{prop}{Proposition}[section]
\newtheorem{thm}[prop]{Theorem}

\newtheorem{conj}[prop]{Conjecture}
\newtheorem{cla}[prop]{Claim}

\newtheorem{rem}[prop]{Remark}

\title[]{Two-dimensional Anisotropic KPZ growth  and  limit shapes}
\author{Alexei Borodin} \address{Department of Mathematics, MIT, Cambridge, USA, and Institute for Information
Transmission Problems, Moscow, Russia. E-mail: {\tt
    borodin@math.mit.edu}}
\author{Fabio Toninelli} \address{Univ Lyon, CNRS, Universit\'e Claude Bernard Lyon 1, UMR5208, Institut Camille Jordan, F-69622 Villeurbanne, France. E-mail: {\tt
    toninelli@math.univ-lyon1.fr}} \date{}

\begin{document}

\maketitle

\begin{abstract}
  A series of recent works focused on two-dimensional interface growth
  models in the so-called Anisotropic KPZ (AKPZ) universality class,
  that have a large-scale behavior similar to that of the
  Edwards-Wilkinson equation.  In agreement with the scenario
  conjectured by D. Wolf \cite{Wolf}, in all known AKPZ examples the
  function $v(\rho)$ giving the growth velocity as a function of the
  slope $\rho$ has a Hessian with negative determinant (``AKPZ
  signature''). While up to now negativity was verified model by model
  via explicit computations, in this work we show that it actually has
  a simple geometric origin in the fact that the hydrodynamic PDEs
  associated to these \emph{non-equilibrium} growth models preserves
  the Euler-Lagrange equations determining the macroscopic shapes of
  certain \emph{equilibrium} two-dimensional interface models.  In the
  case of growth processes defined via dynamics of dimer models on
  planar lattices, we further prove that the preservation of the
  Euler-Lagrange equations is equivalent to
  harmonicity of  $v$ with respect to  a natural complex structure.  
\end{abstract}

\section{Introduction}
In this work we consider stochastic interface growth models,
\emph{i.e.,} irreversible Markov dynamics of $d$-dimensional height functions
$\{h_x\}_{x\in \mathbb Z^d}$. For reasons that we explain below, we
are especially interested in the case $d=2$. We refer the reader,
\emph{e.g.,} to \cite{BarSta} for physical motivations: propagation of fronts
in combustion, crystal growth, etc.  There are various quantities of
physical interest related to the large-scale evolution of such
interfaces. An obvious one is the asymptotic speed of growth $v$,
\emph{i.e.,} the average increase of interface height in unit time, for large
times. Actually, the speed of growth is in general a slope-dependent
quantity $v=v(\rho)$, where $\rho\in \mathbb R^d$ is the 
local average interface slope. Other natural quantities are fluctuation exponents.  In
the long-time limit $t\to\infty$, the law of the height gradients
$\{h_x(t)-h_y(t)\}_{x,y\in \mathbb Z^d}$ of an initially flat profile
of slope $\rho$ is expected to converge to a (non-reversible) stationary state
$\pi_\rho$.  The standard deviation of the height difference
$(h_x-h_{y})$ at stationarity is expected to behave like
$const.\times(1+|x-y|^\alpha)$ for large $|x-y|$, with $\alpha$ the
\emph{roughness exponent}. Similarly, the \emph{growth exponent}
$\beta$ is defined\footnote{More precisely, one
  should look at the standard deviation of $h_{x(t)}(t)-h_x(0)$ with
  $x(t)=x-D v(\rho)t$, where $D v(\rho)\in\mathbb R^d$ is the
  differential of $v(\cdot)$ computed at $\rho$. In other words, to
  correctly define the growth exponent one has to choose a reference
  frame moving along the line $(x(t),t)$ in
  space-time, that is the characteristic line of the PDE \eqref{eq:evoluz}. } so that the standard deviation of $(h_x(t)-h_x(0))$
grows like $t^\beta $ for large $t$.

The two-dimensional case $d=2$ is particularly interesting: it
is expected that there is a rich interplay between the convexity
properties of $v$ and the triviality or not of the critical exponents,
where triviality means that they coincide with those of the linear
Edwards-Wilkinson (EW) equation, recalled below. (Let us recall that
$\alpha_{EW}(d)=(2-d)/2,\beta_{EW}(d)=(2-d)/4$ and that for $d=2$
fluctuation growth in space of time for the EW equation is
logarithmic.)  The following is expected:
\begin{conj}
  \label{conj}
 If the
Hessian matrix $D^2v(\rho)$ of $v(\rho)$ has two eigenvalues of the same
sign, \emph{i.e.,} if $\det(D^2v(\rho))>0$, then
$\alpha\ne\alpha_{EW}(2)=0, \beta\ne\beta_{EW}=0$, while equalities
hold if $ \det(D^2 v(\rho))\le 0$. 
\end{conj}
In the former case the growth model is said to belong to the Isotropic
KPZ class and to the Anisotropic KPZ (AKPZ) class in the latter. This
situation is in contrast with that of $d=1$ growth models where it is
believed \cite{KPZ}, and mathematically proven in several concrete
examples \cite{KPZrev}, that $ \beta=1/3\ne\beta_{EW}(1)=1/4$ as soon
as $v''(\rho)\ne0$.

Conjecture \ref{conj} originated from the work \cite{Wolf} by D. Wolf,
who studied the large-scale behavior of the two-dimensional KPZ equation
\begin{eqnarray}
  \label{eq:KPZ2}
  \partial_t h(x,t)=\Delta h(x,t)+\lambda(\nabla h(x,t),H \nabla h(x,t))+ \dot W(x,t), \quad x\in \mathbb R^2, 
\end{eqnarray}
via (perturbative in $\lambda$) Renormalization Group (RG)
arguments. In this equation, $\dot W$ is a (suitably regularized)
space-time white noise, $\lambda$ is a real parameter tuning the
strength of the non-linearity, and $H$ is a real symmetric $2\times 2$
matrix. This equation reduces to the EW equation when $\lambda=0$, and
it is believed to capture the large-scale behavior of the growth model
if $H$ is chosen to be $D^2 v(\rho)$ as above. The outcome of Wolf's
work is that for $\det(H)>0$ the non-linearity is relevant (\emph{i.e.,}
the large-scale properties of the solution of \eqref{eq:KPZ2} differ
from those of the EW equation as soon as $\lambda\ne0$), while if
$\det(D^2 v(\rho))\le 0$ and $\lambda$ is small enough then the non-linearity is irrelevant.

Wolf's predicted relation between the sign of $\det(D^2 v(\rho))$ and
the non-triviality of the exponents $\alpha,\beta$ has been
successfully tested numerically both on the KPZ equation itself
\cite{HHA92} and on microscopic growth models \cite{TFW}. In
particular, for models in the Isotropic KPZ class, exponents
$\alpha,\beta$ are known with large numerical precision and they seem
to be universal \cite{Num2}. There are also rigorous results about
some specific AKPZ growth models, that we briefly review in Section
\ref{sec:mathAKPZ}.  Apart from Wolf's RG computations, however, we
are not aware of any sound theoretical, let alone rigorous, argument
supporting Conjecture \ref{conj}. Even for those AKPZ growth models of
Section \ref{sec:mathAKPZ} for which logarithmic fluctuation growth
can be proved, the presently existing proofs of
$\det (D^2 v(\rho))\le 0$ are based on explicit computations that give
no intuition on the underlying mechanism. The goal of the present work
is to shed new light on this issue.

Before coming to that, let us mention that there are various growth
processes for which one can prove, via a sub-additivity argument
\cite{Seppa}, that the speed $v$ is a convex function of $\rho$. This
includes the cube-stacking model of \cite{TFW}, that is equivalent to
Glauber dynamics of interfaces of the three-dimensional Ising model
at zero temperature and positive magnetic field.  Convexity of $v$
makes such models natural candidates for representatives of the
Isotropic KPZ universality class. However, apart from refined
numerics, nothing is known rigorously on their stationary states and their
fluctuation exponents.  As we argue at the end of Section
\ref{sec:risultati}, we have reasons to believe that their stationary
states $\pi_\rho$ are not of Gibbs type.

\subsection{Previous results on AKPZ models}
\label{sec:mathAKPZ}
 Here we briefly review
some mathematical results on AKPZ growth models; see also \cite{Ticm}
for a recent review.  The first result we are aware of in this
direction is the study  \cite{PS} of the
Gates-Westcott evolution \cite{GW}, where space is continuous in one
direction and discrete in the other. This growth model can be seen as
a collection of mutually interacting, one-dimensional PolyNuclear
Growth models. The authors of \cite{PS} observed that evolution
preserves a family of translation-invariant Gibbs measures that can be
described in the free-fermionic language and, therefore, have
determinantal correlations. This allowed them to prove that stationary
fluctuations are logarithmic (whence $\alpha=0$) and to compute the
speed $v(\rho)$, which by direct inspection turns out to satisfy
$\det(D^2 v(\rho))<0$.

The other AKPZ growth models we mention below have the common feature
of being defined in terms of dimer coverings of bipartite planar
lattices $G$ (in particular, the honeycomb lattice and the square
grid) or, equivalently, in terms of random tilings of the plane. It is a
well-known fact \cite{Kenyonnotes} that dimer coverings of $G$ are in
bijection with integer-valued height functions defined on the dual
graph $G^*$. See, \emph{e.g.,} Fig. \ref{fig:mapping} for lozenge tilings.
\begin{figure}
\begin{center}
\includegraphics[height=5cm]{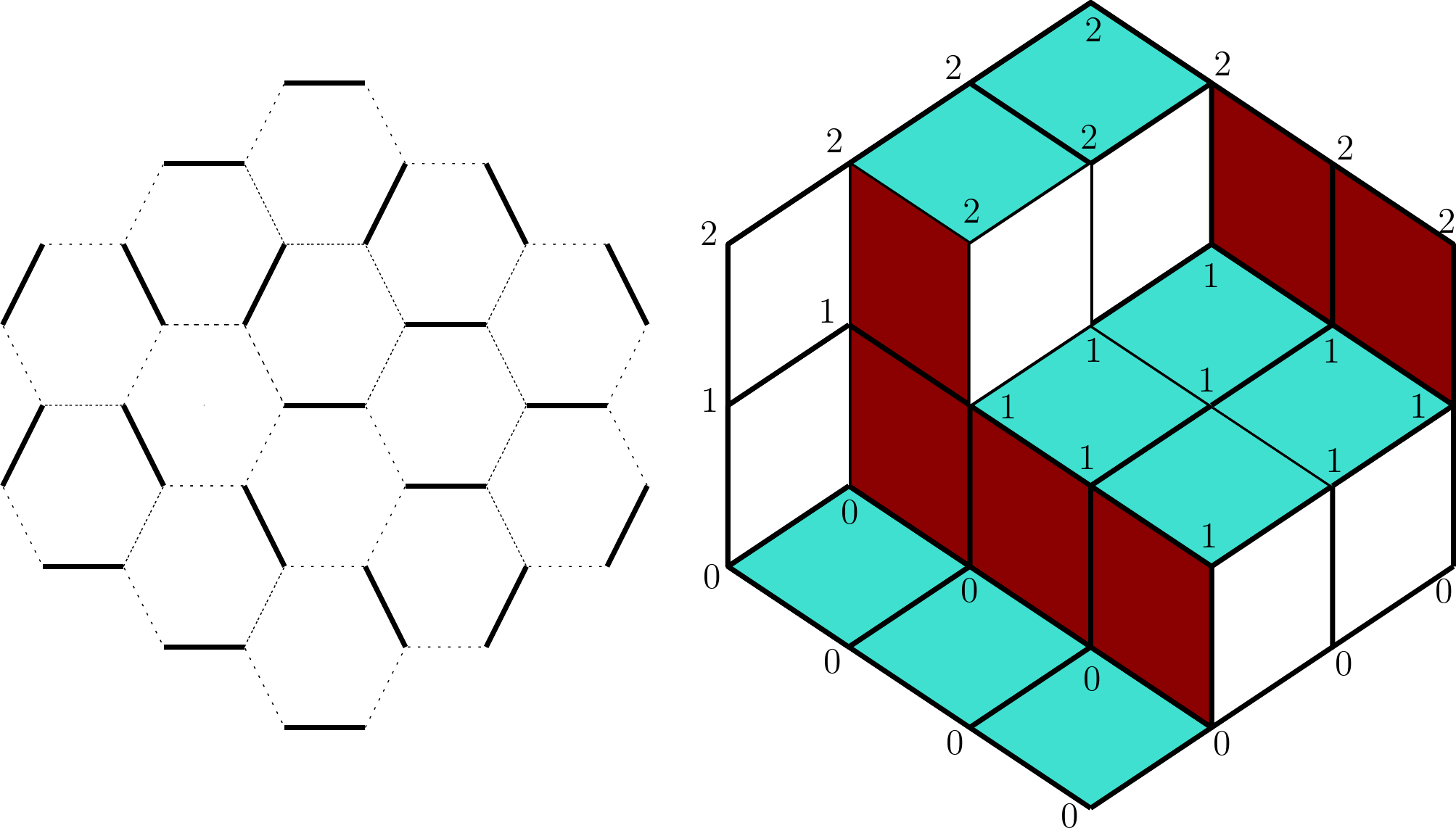}
\caption{A dimer covering of a domain of  the honeycomb lattice (left) and the
  corresponding lozenge tiling (right). Near each lozenge vertex is
  given the height of the interface w.r.t. the horizontal plane. }
\label{fig:mapping}
\end{center}
\end{figure}

A general mechanism for obtaining AKPZ two-dimensional growth
processes was suggested in \cite{BF08}, where the case of a particular
continuous-time dynamics on the lozenge tilings of a half-plane was
considered in detail as the main example (see also \cite{BFjstat}). In
this case, logarithmic growth of fluctuations in time and space was
proven (whence $\alpha=\beta=0$), the speed was computed, and the AKPZ
signature $\det(D^2 v(\rho))<0$ was verified.  The growth models
generated by the procedure of \cite{BF08} have the property of
preserving families of Gibbs measures in certain unbounded regions of
the plane.  One can also prove that full-plane translation invariant
Gibbs measures $\pi_\rho$ are stationary \cite{T15} (this requires
extra non-trivial work because these dynamics involve unbounded
particle jumps, which might cause the process to be ill-defined).  A
particular modification designed to preserve Gibbs measures inside
bounded regions (hexagonal regions for lozenge tilings) was also
suggested at the same time \cite{BG08}.  Other (discrete time)
dynamics for lozenge tilings that fit into the formalism of
\cite{BF08} were considered in \cite{BBO}; their speed of growth
$v(\rho)$ was computed there, and one can check that
$\det(D^2 v(\rho))< 0$. Let us mention that an exposition of the
general construction of \cite{BF08} in terms of the formalism of Schur
processes is available in \cite{B10} and another example for dimers on a square-hexagon lattice can be found in \cite{BF15}.

Another two-dimensional dynamics that recently attracted attention in
the context of AKPZ growth model is the \emph{domino shuffling
  algorithm}. While this evolution fits into the general framework of
\cite{BF08}, it was actually introduced much earlier in
\cite{Aztec1,Aztec2}, as an algorithm that perfectly samples domino
tilings of special domains of the plane called ``Aztec
diamonds''. However, the algorithm can be defined in more general
domains: in particular, in the whole plane.  Thanks to the
above-mentioned mapping between a domino tiling and its discrete
height function, the shuffling algorithm can be viewed as a
two-dimensional growth model.  In the whole plane, it preserves
translation-invariant Gibbs measures of domino tilings.  In the recent
work \cite{CT}, the shuffling algorithm was studied from the point of
view of interface growth, and it was proven to belong to the AKPZ
class, both in terms of vanishing of the growth exponents and in the
sense that $\det(D^2 v(\rho))<0$. An interesting fact observed in
\cite{CT} is that, if the underlying dimer model is given non-uniform
but periodic weights on $\mathbb Z^2$, the speed function $v(\rho)$
can be singular (non-differentiable) at isolated values of the slope
$\rho$ (the case of weights with period $2$ in both space directions is
worked out explicitly in \cite{CT}). At those slopes, growth exponents
$\alpha,\beta$ are still zero but the fluctuation variance is bounded
in time and space.

We conclude this (incomplete) list of results by mentioning two more
examples. The first one is an AKPZ growth model based on domino
tilings in the plane, defined in Section 3.1 of \cite{T15}, that at
first sight does not seem to fit into the general formalism of
\cite{BF08}. The proof of logarithmic growth fluctuations, implying
$\alpha=\beta=0$, as well as the computation of the speed of growth
and a direct verification that $\det(D^2 v(\rho))<0$ can be found in
\cite{CFT}. The second one is the so-called $q$-Whittaker process,
introduced in \cite{Macdo} in the wider context of Macdonald
processes. The $q$-Whittaker dynamics reduces to the growth process of
\cite{BF08,BFjstat} in the limit $q\to0$ and to a Gaussian growth
model (that has the same asymptotic large-scale behavior as the EW
equation) in the $q\to1$ limit \cite{q11,q12}. For every intermediate
value of $q$, it preserves certain Gibbs measures, both on the torus
\cite{CorT} and in certain unbounded domains of the plane. Logarithmic
growth of fluctuations is not proven but  conjectured.

For some of the models mentioned in this section, a hydrodynamic limit
has been proven: the height profile, rescaled as
$\epsilon h_{\epsilon^{-1}x}(\epsilon^{-1}t)$, converges as
$\epsilon\to0$ to a non-random limit profile $\bar h(x,t)$, that solves a
first-order non-linear PDE of Hamilton-Jacobi type. See for instance
\cite[Th. 1.2]{BF08} and \cite{LT} for a lozenge tiling dynamics, and
\cite{XZ} for the domino shuffling algorithm. Most likely, a similar
convergence holds for the other models as well.

\subsection{Preservation of Gibbs measures and of equilibrium shapes}
\label{sec:risultati}

A feature common to most (and possibly all) of the models cited in the
previous section is that they admit certain translation-invariant
Gibbs measures as stationary, but not reversible, measures for the
height gradients (the law of the height itself is not stationary,
since there is a non-zero average growth).  In almost all of these
models, Gibbs states have a determinantal (or free-fermion) structure,
but this is not a general feature of AKPZ models: the $q$-Whittaker
process is \emph{not} free-fermionic, for instance. The potential
associated to these Gibbs measures is local and completely explicit:
in several cases, for example for the lozenge dynamics of
\cite{BF08,BFjstat}, the Gibbs property means simply that,
conditionally on the interface configuration $h_{\Lambda^c}$ 
outside any finite domain $\Lambda$ of the lattice, the configuration
inside $\Lambda$ is uniformly distributed among all possible
configurations compatible with $h_{\Lambda^c}$. Actually, the
general construction of \cite{BF08} guarantees that these growth
models preserve also non-translation invariant Gibbs measures, in the
sense that if the initial condition is one such measure, then the law
of the interface at a later time is another such measure (in general,
a different one). See \cite[Fig. 1.3]{BF08} for a configuration
sampled from one of these measure for lozenge tilings.  Another example of
preservation of non-translation-invariant Gibbs measure is provided by
the domino shuffling algorithm: if at time $0$ the configuration is
sampled from the Gibbs distribution in an Aztec diamond of size
$n\times n$, at time $T$ it is distributed according to the Gibbs
distribution in an Aztec diamond of size $(n+T)\times (n+T)$.

Preservation of Gibbs measures has a direct consequence on the
hydrodynamic PDE of these growth models. To explain this, let us
forget dynamics for a moment. According to the general principle of
statistical mechanics, on large scales (\emph{i.e.,} rescaling height
as $\epsilon h_{\epsilon^{-1}x}$ and letting $\epsilon\to0$), the
height function sampled from a Gibbs measure concentrates around a
non-random profile $\bar h$, the \emph{equilibrium shape}. This
profile is determined by minimizing a surface tension functional
$\int \sigma(\nabla \bar h)dx$, with suitable boundary
conditions. More explicitly, the function $\bar h(x)$ satisfies the
Euler-Lagrange equation associated to such functional; due to
convexity (see, \emph{e.g.,} \cite[Ch. 4]{Sheffield}) of the surface tension $\sigma(\cdot)$, this is a
second-order PDE of elliptic type. Back to the growth models: the fact
that the above AKPZ evolutions preserve Gibbs measures translates into
the fact that their hydrodynamic PDE preserves the Euler-Lagrange
equation. Namely, if the initial profile satisfies the Euler-Lagrange
equation, then so does it at later times $t$.

The main contribution of the present work, Theorem \ref{thm1}, is an
elementary argument that shows that, roughly speaking, if the
hydrodynamic PDE of a two-dimensional growth model preserves the Euler-Lagrange equation of
\emph{some} equilibrium interface model, then the speed of growth
function $v$ satisfies $\det (D^2 v(\rho))\le0$.  In the particular
case where the growth model is defined in terms of lozenge tilings, or,
more generally, in terms of a planar dimer model we show a stronger
result (Theorems \ref{thm4} and \ref{thm5}): the hydrodynamic PDE
preserves the Euler-Lagrange equation of the dimer model if and
only if the speed $v$ is a harmonic function with respect to a certain
natural complex-valued variable defined in terms of the interface
slope \cite{KO,Kcmp}.

We conclude this introduction with a couple of observations.  First,
we emphasize that Theorem \ref{thm1} implies that the hydrodynamic PDE
of Isotropic KPZ models, like the cube-stacking dynamics, cannot
preserve the Euler-Lagrange equation of an equilibrium interface model
with convex surface tension. This suggests the possibility that the
stationary states of such growth processes are \emph{not} of Gibbs
type.  The occurrence of non-Gibbs stationary states in irreversible
interacting particle systems is believed to be a generic feature
\cite{LS} but it has been proven only in very few examples
(cf. Section 4.5.4 of \cite{nonGibbs} and references therein).
Secondly, our argument might be instrumental in proving
$\det(D^2(\rho))\le 0$ for some AKPZ models for which the explicit
computation of $v$ is not possible. We think, in particular, of the
above-mentioned $q$-Whittaker process, whose invariant (Gibbs)
measures are not of determinantal type and do not allow for explicit
computations. A last remark is that, while in equilibrium statistical
mechanics it is usually possible to guess the universality class a model belongs to 
from the symmetries of the Hamiltonian, for growth models it is not
clear how to guess the convexity properties of $v(\cdot)$ from the
definition of the generator.

\section{Euler-Lagrange equation and AKPZ signature}

From this point on, we denote macroscopic profiles
(time-dependent as well as time-independent) by $h$ instead of
$\bar h$, since the microscopic height $h_x$ will not appear in any of
the subsequent arguments.

Let $\sigma:\mathbb R^2\mapsto \mathbb R$ be the surface tension
function of an equilibrium two-dimensional interface model, with
$\sigma(\rho)$ denoting the surface tension at slope $\rho$. We will
denote by $\sigma_{i,j}$ the derivative of $\sigma$ w.r.t. components
$i,j$ of $\rho$, and by $\Sigma(\rho)$ the $2\times 2$ matrix with
\begin{eqnarray}
  \label{eq:Sigma}
\Sigma(\rho)_{i,j}:=\sigma_{i,j}(\rho).  
\end{eqnarray}
Since $\sigma$ is convex, $\Sigma$ is positive definite (possibly not
strictly). In many natural examples (notably the dimer model discussed
in Section \ref{sec:dimer}), the interface configurations are
Lipschitz, and the allowed slopes belong to some convex set $N$ that
in the dimer model setting is called ``Newton Polygon''. In this case,
$\sigma$ equals $+\infty$ outside $N$ and $0$ on the boundary of $N$.

An ``equilibrium shape'' is a 
 height profile  $ h :\mathbb R^2\mapsto \mathbb R$ that locally minimizes the surface tension 
functional
\begin{eqnarray}
  \label{eq:stf}
  F ( \varphi )=\int_{\mathbb R^2} \sigma(\nabla  \varphi )dx.
\end{eqnarray}
(One can also consider the optimization problem restricted to a
sub-domain $U$ of the plane, in which case the boundary height
$h|_{\partial U}$ on the boundary of $U$ is fixed.)  More precisely,
for every finite subset $V\subset \mathbb R^2$ whose boundary is a
smooth simple curve, the minimum of the functional
\begin{eqnarray}
  \int_V \sigma(\nabla \varphi)dx
\end{eqnarray}
over functions on $V$ with boundary height
$\varphi|_{\partial V}\equiv h |_{\partial V}$ 
is realized by  the restriction of $h$ itself to $V$.  At every point
 where $h$ is $C^2$, it satisfies the Euler-Lagrange equation
\begin{eqnarray}
  \label{eq:EL}
\mathcal L[ h ](x):=  \sum_{i,j=1}^2\sigma_{i,j}(\nabla  h (x))\partial^2_{x_i x_j} h(x) =0,
\end{eqnarray}
that is an elliptic PDE.
While  for $d=1$  solutions of the Euler-Lagrange equation are affine, this is not necessarily the case in dimension $d=2$ (and higher).

Now consider a two-dimensional growth model and assume it satisfies a hydrodynamic limit with speed function $v$. Namely, its  random height profile, once rescaled as explained at the end of Section \ref{sec:mathAKPZ}, converges to the  solution of the  deterministic Hamilton-Jacobi equation 
\begin{eqnarray}
  \label{eq:evoluz}
  \left\{
  \begin{array}[l]{l}
\partial_t  h(x,t) =v(\nabla  h(x,t) )\\ h(x,0)=h_0(x)    
  \end{array}
  \right.
\end{eqnarray}
where $v(\cdot)$ is some function defined on the set $N$ of allowed
slopes.
Non-linear Hamilton-Jacobi equations like \eqref{eq:evoluz} are
well-known to develop singularities (discontinuities of $\nabla h$) in
finite time, and there is not a unique way of defining a weak solution
after the time of appearance of singularities. However, the physically
relevant solution defined for all times is the so-called viscosity
solution \cite{Viscosity}, that can be obtained by adding
$\nu\Delta h(x,t)$ to $v(\nabla h(x,t))$ and then sending
$\nu\to 0^+$.  We let $ h (t)$ denote the viscosity solution of
\eqref{eq:evoluz} at time $t$. We assume henceforth that
$v(\cdot):N \mapsto \mathbb R$ is smooth enough so that the viscosity
solution exists and is unique.

Our first result says, roughly speaking, that if the PDE
\eqref{eq:evoluz} preserves the Euler-Lagrange equation \eqref{eq:EL}, then
the determinant of the Hessian of $v$ cannot be positive. More
precisely:
\begin{thm}
  \label{thm1}
  Let $\rho$ be a slope in the interior of $N$ where both $v(\cdot)$
  and $\sigma(\cdot)$ are $C^2$ differentiable and where $\Sigma(\rho)$ in \eqref{eq:Sigma} is
  strictly positive definite. Assume that there exists an equilibrium
  shape $h$ and a point $x\in \mathbb R^2$ where $h(\cdot)$ is $C^2$
  differentiable, and such that:
\begin{itemize}
  
  \item $\nabla h (x)=\rho$ and the Hessian of $ h $ at $x$ is not zero;

  \item one has
    \begin{eqnarray}
      \label{eq:dL0}
\mathcal L[h(t)](y)=0
    \end{eqnarray}
    for $y$ in a neighborhood of $x$ and $t$ sufficiently small.
  \end{itemize}
  Then, $\det(D^2 v(\rho))\le 0$.
\end{thm}

\begin{proof}
  For sufficiently small times the PDE \eqref{eq:evoluz} can be solved
  by the method of characteristics, and the solution is $C^2$ in the
  neighborhood of $x$.  Take $x=x(t)$ that runs along the
  characteristic line started at $x$, \emph{i.e.,}
\begin{eqnarray}
  \label{eq:char}
  \frac{d}{dt} x(t)=-D v(\nabla h(x(t),t) ), \qquad x(0)=x,
\end{eqnarray}
with $Dv$ the gradient of $v$. Recall that $\nabla  h $ is constant along the characteristic lines, so that $d x(t)/dt=-D v(\nabla h(x))$ is constant.
Call
\begin{eqnarray}
  \label{eq:R}
  R(x):=\left.\frac{d}{dt}\mathcal L[ h (t)](x(t))\right|_{t=0},
\end{eqnarray}
that equals zero by the assumption of the theorem.

Since $\nabla h $  is constant along characteristics, Eq. \eqref{eq:R} gives
\begin{eqnarray}
  \label{eq:EL3}
  R(x)=\sum_{i,j=1}^2\sigma_{i,j}(\nabla  h (x))\left.\frac{d}{dt}\partial^2_{x_i x_j} h (x(t),t)\right|_{t=0}.
\end{eqnarray}
Using the chain rule for derivatives together with the definition \eqref{eq:char} of characteristic lines
we get
\begin{multline}
  \left.  \frac{d}{dt}\partial^2_{x_i x_j} h (x(t),t)\right|_{t=0}
= \left. \sum_{\ell,k=1}^2 D^2_{\ell, k} v(\nabla  h (x))\partial^2_{x_j x_\ell} h (x(t),t)\partial^2_{x_i x_k} h (x(t),t)\right|_{t=0},
\end{multline}
where \[D^2_{\ell,k}v(\rho):= \frac{\partial^2}{\partial \rho_\ell\partial\rho_k}v(\rho).\]
Altogether, we have obtained
\begin{eqnarray}
  \label{eq:obt}
0=R(x)=  \sum_{\ell,k=1}^2[D^2_{\ell, k} v(\nabla  h (x))] A_{\ell,k}(x)
\end{eqnarray}
with
\begin{eqnarray}
  \label{eq:A}
  A_{\ell,k}(x)=\sum_{i,j=1}^2\sigma_{i,j}(\nabla h (x))\;\partial^2_{x_i x_\ell} h (x) \;\partial^2_{x_j x_k} h (x) .
\end{eqnarray}

For lightness of notation, call $M:=D^2 v(\rho)$ the Hessian matrix of
$v$ computed at $\rho$ and write $\sigma_{i,j}$ instead of
$\sigma_{i,j}(\nabla h(x))$.  Assume by contradiction that $\det(M)>0$
and that $M$ is positive definite (the argument works analogously if
$\det(M)>0$ and $M$ is negative definite), so that we can write it as
$M=\sqrt M \sqrt M$.  Also call $u_k, k=1,2,$ the vector
$(\partial^2_{x_1 x_k} h(x) ,\partial^2_{x_2 x_k} h(x) )$.  We rewrite
\eqref{eq:obt} as
\begin{equation}
  \nonumber
0=  \sigma_{1,1}\|\sqrt M u_1\|^2 +  \sigma_{2,2}\|\sqrt M u_2\|^2+2\sigma_{1,2}(\sqrt M u_1,\sqrt M u_2).
\end{equation}
Since the matrix $\Sigma$ is strictly positive
definite, \[|\sigma_{1,2}|<\sqrt{\sigma_{1,1}\sigma_{2,2}}\quad \textrm{and} \quad 
\sigma_{1,1},\sigma_{2,2}>0.\]  By the assumption that the Hessian of
$h$ at $x$ is non-zero, either $u_1$ or $u_2$ is non-zero. If one of
them, say $u_2$, is zero, then one gets
\begin{eqnarray}
  \sigma_{1,1}\|\sqrt M u_1\|^2=0
\end{eqnarray}
so that $M$ is not strictly positive definite, which is a contradiction. 
Also, if either $\sqrt Mu_1$ or $\sqrt M u_2$ is zero, then again $\det(M)=0$ which contradicts the hypothesis $\det(M)>0$.
From now on we can therefore assume that $\sqrt M u_1,\sqrt M u_2\ne0$.
 Also, $\sqrt M u_1$ cannot be orthogonal to $\sqrt M u_2$; otherwise
\[
0=\sigma_{1,1}\|\sqrt M u_1\|^2+\sigma_{2,2}\|\sqrt M u_2\|^2,
  \]
  which is not possible since $\sigma_{1,1}>0,\sigma_{2,2}>0$. In the
  remaining case where $(\sqrt M u_1,\sqrt M u_2)\ne 0$,
\begin{multline}
  0\ge  \sigma_{1,1}\|\sqrt M u_1\|^2+  \sigma_{2,2}\|\sqrt M u_2\|^2-2|\sigma_{1,2}|\, |(\sqrt M u_1,\sqrt M u_2)|\\
  > \sigma_{1,1}\|\sqrt M u_1\|^2+  \sigma_{2,2}\|\sqrt M u_2\|^2\\-2\sqrt{\sigma_{1,1}\|\sqrt M u_1\|^2}
  \sqrt{\sigma_{2,2}\|\sqrt M u_2\|^2}\ge0 
\end{multline}
that is a contradiction. 
Altogether, it is not possible that $\det(M)>0$.
\end{proof}

\section{Dimer model: Burgers equation and harmonicity of $v$}
\label{sec:dimer}
In this section, we assume that $\sigma$ is the surface tension of a
dimer model on a bipartite, planar, periodic graph.  We refer,
\emph{e.g.,} to \cite{Kenyonnotes} for an introduction to dimer models
and to \cite{KOS} for the definition of their ergodic Gibbs measures;
we recall here a minimum of basic facts.  Given a planar, bipartite,
periodic weighted graph $G=(V,E)$ (positive weights are associated to
edges), dimer configurations are perfect matchings of $G$, and a
canonical construction allows to associate to a dimer configuration an
integer-valued height function on faces of $G$ (the map is bijective
up to a global additive constant for the height). The possible slopes
of the height function belong to a convex polygon $N$, called ``Newton
polygon'', whose vertices have integer coordinates. For every slope
$\rho$ in the interior of $N$ there exists a unique translation
invariant, ergodic Gibbs measure $\pi_\rho$ on dimer coverings, with
average slope $\rho$.  Conditionally on the dimer configuration
outside any finite sub-graph $\tilde G$, $\pi_\rho$ assigns to any
admissible dimer configuration in $\tilde G$ a probability
proportional to the product of weights of edges occupied by
dimers. Slopes $\rho$ that are in the interior of $N$ and whose
coordinates are not both integer are called ``liquid slopes''. In this
case, $\pi_\rho$ is called a ``liquid phase'' and it is known that
dimer correlations decay polynomially and the height field behaves
like a log-correlated massless Gaussian field on large scales. If
$\rho$ is in the interior of $N$ but has integer coordinates, then for
generic choice of the edge weights the measure $\pi_\rho$ is a
``gaseous'' or ``smooth'' phase, with exponentially decaying
correlations and $O(1)$ height fluctuations\footnote{For special edge
  weights, the measure $\pi_\rho$ can be liquid (with polynomially
  decaying correlations) even for integer-valued slopes \cite{KOS}.
}. A crucial object for the dimer model is the
so-called characteristic polynomial $P(z,w)$: this is a Laurent
polynomial in $z,w\in \mathbb C$, that is  associated to the
weighted graph $G$. For instance, the surface tension $\sigma(\cdot)$  is 
obtained as the Legendre transform with respect of $B=(B_1,B_2)\in\mathbb R^2$ of
\begin{eqnarray}
  \label{eq}
  \frac1{(2\pi i)^2}\int_{|z|=e^{B_1}}\frac{dz}z \int_{|w|=e^{B_2}}\frac{dw}w\log P(z,w).
\end{eqnarray}

For the dimer model, it is known that the Euler-Lagrange equations
\eqref{eq:EL} for equilibrium shapes $h$ can be expressed via a
\emph{first-order} PDE, that was called \emph{complex Burgers
  equation} in \cite{KO}.  Namely, for $x=(x_1,x_2)$, define non-zero
complex numbers $z=z(x),w=w(x)$ via the relations\footnote{Our $z,w$ correspond to the complex conjugates
  of the similarly denoted complex numbers in \cite{KO}.}
\begin{eqnarray}
  P(z,w)=0,\qquad
  \label{eq:nablazw}
  \nabla  h(x) =\frac1\pi(-\arg w,\arg z).
\end{eqnarray}
According to \cite[ Theorem 1]{KO}, at every point $x$ in the liquid
region (\emph{i.e.,} such that  in the neighborhood of $x$, $h$ is $C^1$ with $\nabla h$ a liquid slope), the
Euler-Lagrange equation \eqref{eq:EL} is equivalent to the ``complex
Burgers equation''
\begin{eqnarray}
  \label{eq:Burgzw}
  \frac{z_{x_1}}z+\frac{w_{x_2}}w=0,
\end{eqnarray}
where $f_{x_i}$ denotes the derivative of $f$ w.r.t. $x_i$.
Note that relation \eqref{eq:nablazw} does not really fix $z,w$
because the argument is a multi-valued function; as
explained in \cite{KO}, the statement is that there exists a branch of
the argument for which equality \eqref{eq:nablazw} is satisfied.

Take a point $x_0$ in the liquid
region, where the gradient of $h$ is
\begin{eqnarray}
  \label{eq:rho0z}
\rho=\nabla h (x_0)=\frac1\pi(-\arg w_0,\arg z_0)  
\end{eqnarray}
for some
$z_0,w_0$ such that $P(z_0,w_0)=0$. Locally around $(z_0,w_0)$ we
can 
expand $P(z,w)$ as
\begin{eqnarray}
  P(z,w)=p_1 (z-z_0)+p_2(w-w_0)+O(\|(z,w)-(z_0,w_0)\|^2)
\end{eqnarray}
with $p_1=\partial_zP(z_0,w_0),p_2=\partial_wP(z_0,w_0)$.
As discussed in \cite[Sec. 2.3]{KO}, the ratio $p_1/p_2$ is neither zero nor infinite.
Then, by the implicit function theorem we can locally write $z$ as a differentiable function $z(w)$ or, conversely, write $w$ as a differentiable function $w(z)$.

An identity that will be important in the following is that
\begin{eqnarray}
  \label{eq:idKO}
  \frac{\partial\log |z|}{\partial \nabla_{x_2} h }=  \frac{\partial\log |w|}{\partial \nabla_{x_1} h },
\end{eqnarray}
that follows from \cite[Eq. (12)]{KO}.
We have also
\begin{multline}
  \label{eq:id2}
  \frac{\partial\log |z|}{\partial \nabla_{x_1} h }=   \frac{\partial\log z}{\partial \nabla_{x_1} h },\qquad  \frac{\partial\log |z|}{\partial \nabla_{x_2} h }=   \frac{\partial\log z}{\partial \nabla_{x_2} h }-i\pi\\
  \frac{\partial\log |w|}{\partial \nabla_{x_1} h }=   \frac{\partial\log w}{\partial \nabla_{x_1} h }+i\pi,\qquad
  \frac{\partial\log |w|}{\partial \nabla_{x_2} h }=   \frac{\partial\log w}{\partial \nabla_{x_2} h },  
\end{multline}
that follow from \eqref{eq:nablazw}.

Relations $w=w(z)$  and \eqref{eq:nablazw} give a bijection between
$z$ in a neighborhood of $z_0$ and the interface slope
$\nabla h =(\nabla_{x_1} h ,\nabla_{x_2} h )$ in a neighborhood of
$\nabla h (x_0)$. Therefore, we can write locally (\emph{i.e.,} for $\nabla h$
close to $\nabla h(x_0)$)  the speed $v(\nabla h )$ as some function
$f(z(\nabla h))$.
\begin{rem}
  Globally, the mapping from $z$ to the slope is multi-valued, since
  for a given $z$ there may be many $w$ satisfying $P(z,w)=0$.  
  Notable exceptions are the dimer model on the honeycomb lattice and the  square grid with
  translation-invariant weights.  We will briefly come back to these special cases in Section \ref{sec:lozenge}.
  See instead \cite{CT} for a growth model in a case where the solution of $P(z,w)=0$ has several branches.
\end{rem}

Let $ h = h (x,t)$ be the viscosity solution of the PDE
\eqref{eq:evoluz}. For  $t$ small and $x$ around $x_0$ the solution
is smooth, and we let $ z =z(x,t)$ be defined by \eqref{eq:nablazw}
with $h$ replaced by $h(t)$, with a choice of the branch of the
argument so that $z(x_0,0)=z_0$.  Define also
  \begin{eqnarray}
    \label{eq:Delta}
    \Delta= z  w _{x_2}+w z _{x_1}
  \end{eqnarray}
  so that $\Delta=0$ iff
  the complex Burgers equation \eqref{eq:Burgzw} holds.

\begin{thm}
  \label{thm4} Under the same assumptions as for Theorem \ref{thm1},
  assume in addition that $\sigma$ is the surface tension function   of a
  dimer model and $\rho$ is a slope of the liquid phase.
 Then, 
  $f( \cdot )$ is a harmonic function at $z_0$, \emph{i.e.,}
  \begin{eqnarray}
    \label{z0}
   \frac{\partial^2 f( z)}{\partial \Re z ^2}+\frac{\partial^2 f( z )}{\partial \Im z ^2}=0
  \end{eqnarray}
  for $z=z_0$, where $z_0$ is related to $\rho $ as in \eqref{eq:rho0z}.
\end{thm}
Conversely, we have:
\begin{thm}\label{thm5}
  Let $h$ be an equilibrium shape and $h(t)$ be the solution of
  \eqref{eq:evoluz} with initial condition $h$ and $v(\cdot)=f(z(\cdot))$. Assume that $\Delta$
  defined in \eqref{eq:Delta} is differentiable in time and space for
  $x\in A,t\in[0,T]$, for some neighborhood $A$ of $x_0$ and some
  $T>0$.  If $f(z)$ is a harmonic function of $z$, then
  $ h (t)$ satisfies the complex Burgers equation for
  $(x,t)\in A\times [0,T]$.
\end{thm}

The restriction to small times is simply due to the fact that at large
times the solution $h(t)$ might not be pointwise differentiable in
space and time, in which  case $\Delta$ is not well defined.
\begin{rem}
  Together with Theorem \ref{thm1}, Theorem \ref{thm4} implies the
  following: if $v(\nabla h)=f(z(\nabla h))$ with $f$ harmonic at every $z$
  corresponding to a liquid slope, then $\det(D^2v)\le 0$.  It should
  be possible to obtain this also directly by calculus, but we find
  that the path going through Theorem \ref{thm1} is more illuminating.
\end{rem}

  \begin{proof}[Proof of Theorems \ref{thm4} and \ref{thm5}]
    We start by writing
    \begin{eqnarray}
      \label{eq:dlogz}
      \partial_t\log z=\partial_t\log |z|+i\partial_t\arg(z)=\partial_t\log |z|+i\pi\partial_{x_2} f,
    \end{eqnarray}
    where we used $\arg(z)=\nabla_{x_2}h$ and
    $\partial_th=f(z(\nabla h))$.  Also,
    \begin{multline}
      \label{eq:daz}
      \partial_t\log |z|=\frac{\partial\log |z|}{\partial
        \nabla_{x_1} h }\partial_{x_1} f+\frac{\partial\log
        |z|}{\partial \nabla_{x_2} h }\partial_{x_2} f\\=
      \frac{\partial\log |z|}{\partial
        \nabla_{x_1} h }\partial_{x_1} f+\frac{\partial\log
        |w|}{\partial \nabla_{x_1} h }\partial_{x_2} f
      \\
=\frac{\partial f}{\partial \Re z}\Re\left[z_{x_1}\frac{\partial\log |z|}{\partial
        \nabla_{x_1} h }+z_{x_2} \frac{\partial\log
        |w|}{\partial \nabla_{x_1} h }\right]+ \frac{\partial f}{\partial \Im z}\Im\left[z_{x_1}\frac{\partial\log |z|}{\partial
        \nabla_{x_1} h }+z_{x_2} \frac{\partial\log
        |w|}{\partial \nabla_{x_1} h }\right]\\
    = \frac{\partial f}{\partial \Re z}\Re\left[z_{x_1}\frac{\partial\log z}{\partial
          \nabla_{x_1} h } + z_{x_2}\left(i\pi+\frac{\partial\log
            w}{\partial \nabla_{x_1} h }\right) \right]
 \\ +  \frac{\partial f}{\partial \Im z}\Im \left[z_{x_1}\frac{\partial\log z}{\partial
          \nabla_{x_1} h } + z_{x_2}\left(i\pi+\frac{\partial\log
            w}{\partial \nabla_{x_1} h }\right) \right].
    \end{multline}
    In the second step we used \eqref{eq:idKO}, in the third the identities
  \begin{eqnarray}
    \label{eq:using2}
    \partial_{x_1} f( z )&=& \Re z _{x_1}\frac{\partial f( z )}{\partial\Re z }+\Im z _{x_1}\frac{\partial f( z )}{\partial\Im z },\\
    \nonumber
    \partial_{x_2} f( z )&=& \Re z _{x_2}\frac{\partial f( z )}{\partial\Re z }+\Im z _{x_2}\frac{\partial f( z )}{\partial\Im z },
  \end{eqnarray}
    and in the last \eqref{eq:id2}.
    
    Note that, writing $z=z(w)$ and recalling the definition
    \eqref{eq:Delta} of $\Delta$, one has
\begin{eqnarray}
  \label{eq:zy}
  z_{x_2}=z'(w)w_{x_2}=z'(w)\left[\frac\Delta z-\frac wz z_{x_1}\right],
\end{eqnarray}
so that
\begin{eqnarray}
  \label{eq:tgw}
  \left[z_{x_1}\frac{\partial\log z}{\partial
          \nabla_{x_1} h } + z_{x_2}\left(i\pi+\frac{\partial\log
            w}{\partial \nabla_{x_1} h }\right) \right]=i \pi z_{x_2}+\frac\Delta z z'(w)\frac{\partial\log
            w}{\partial \nabla_{x_1} h },
\end{eqnarray}
because
\[
z\partial_{\nabla_{x_1} h }\log z-z'(w)w\partial_{\nabla_{x_1} h }\log w=\partial_{\nabla_{x_1} h } z-\partial_{\nabla_{x_1} h } z=0.
  \]
  Together with \eqref{eq:dlogz}, \eqref{eq:daz}, \eqref{eq:using2}, we obtain
  \begin{eqnarray}
    \label{eq:dlogz2}
    \partial_t\log z=\pi\left[i\partial_{x_2} f-\Im z_{x_2}\frac{\partial f}{\partial \Re z}+\Re z_{x_2} \frac{\partial f}{\partial \Im z}\right]+a= 2i\pi z_{x_2}\frac{\partial f}{\partial z}+a
  \end{eqnarray}
  with $\partial_z f=(1/2)(\partial_{\Re z}-i\partial_{\Im z})f$ and
  \begin{eqnarray}
    \label{eq:a}
    a=\frac{\partial f}{\partial \Re z}\Re\left[\frac\Delta z z'(w)\frac{\partial\log w}{\partial \nabla_{x_1} h }\right]
+    \frac{\partial f}{\partial \Im z}\Im\left[\frac\Delta z z'(w)\frac{\partial\log w}{\partial \nabla_{x_1} h }\right].
  \end{eqnarray}
  Similarly, one gets
  \begin{eqnarray}
\nonumber
    \partial_t\log w=\pi\left[-i\partial_{x_1} f+\Im z_{x_1}\frac{\partial f}{\partial \Re z}-\Re z_{x_1} \frac{\partial f}{\partial \Im z}\right]+b= -2i\pi z_{x_1}\frac{\partial f}{\partial z}+b
  \end{eqnarray}
with 
\begin{eqnarray}
    \label{eq:b}
  b=\frac{\partial f}{\partial \Re z}\Re\left[\frac\Delta z z'(w)\frac{\partial\log w}{\partial\nabla_{x_2} h }\right]
  +    \frac{\partial f}{\partial \Im z}\Im\left[\frac\Delta z z'(w)\frac{\partial\log w}{\partial\nabla_{x_2} h } \right].
  \end{eqnarray}

From \eqref{eq:Delta} we see that
  \begin{multline}
    \label{eq:dDt}
    \partial_t \Delta= 
    z w_{x_2} \left[ 2i\pi z_{x_2}\frac{\partial f}{\partial z}+a
      \right]+z \partial_{x_2}\left[w\left(-2i \pi z_{x_1}\frac{\partial f}{\partial z}+b\right)\right]\\+w z_{x_1}\left[-2i\pi z_{x_1}\frac{\partial f}{\partial z}+b\right]+w\partial_{x_1}\left[z\left(2i\pi z_{x_2}\frac{\partial f}{\partial z}+a\right)\right]
 .
\end{multline}
  Computing the derivatives in \eqref{eq:dDt}  and simplifying, one is left in the end with
  \begin{multline}
    \label{eq:dDt2}
    \partial_t \Delta=(a+b)\Delta+w z(a_{x_1}+b_{x_2})+2i \pi(z_{x_2}-z_{x_1})\Delta\frac{\partial f}{\partial z}
   \\ + 2i \pi z w
     \left[z_{x_2} \partial_{x_1}\frac{\partial f}{\partial z} -z_{x_1} \partial_{x_2} \frac{\partial f}{\partial z}\right]\\
    =(a+b)\Delta+w z(a_{x_1}+b_{x_2})+2i\pi (z_{x_2}-z_{x_1})\Delta\frac{\partial f}{\partial z}\\
 -\pi z w\left(\Im z _{x_2}\Re z _{x_1}-\Re z _{x_2}\Im z _{x_1}\right)\times\left(\frac{\partial^2 f( z )}{\partial \Re z ^2}+\frac{\partial^2 f( z )}{\partial \Im z ^2}\right),
  \end{multline}
  that is the main achievement of the computation.
  
  Let us prove Theorem \ref{thm4}. By assumption, $\Delta$ is zero at time zero and $\partial_t\Delta$ is also zero, so that
  \begin{eqnarray}
    \label{eq:st}
    0=z w\left(\Im z _{x_2}\Re z _{x_1}-\Re z _{x_2}\Im z _{x_1}\right)\times\left(\frac{\partial^2 f( z )}{\partial \Re z ^2}+\frac{\partial^2 f( z )}{\partial \Im z ^2}\right).
  \end{eqnarray}
  Note that $(\Im z _{x_1}\Re z _{x_2}-\Re z _{x_1}\Im z _{x_2})$ vanishes
  only if $ z _{x_2}/ z _{x_1}$ is real. On the other hand,  \eqref{eq:zy} implies that if 
  $\Delta=0$ then
  \begin{eqnarray}
    \label{eq:zy2}
    \frac{z_{x_2}}{z_{x_1}}=-z'(w)\frac wz=-\frac{w P_w(z,w)}{z P_z(z,w)}
  \end{eqnarray}
  and, as shown in \cite[Sec. 2.3]{KO}, the r.h.s. belongs to $\mathbb C\setminus \mathbb R$.
  The claim of the theorem then follows.

  As for Theorem \ref{thm5}, recall from definitions \eqref{eq:a} and
  \eqref{eq:b} that $a,b$ are linear in $\Delta$ and that $f$ is
  harmonic by assumption.  If we consider $ z,w $ as  known functions
  of space and time, then
  one sees that the r.h.s. of \eqref{eq:dDt2} is Lipschitz in
  $\Im\Delta,\Re\Delta$ and its derivatives w.r.t. $x_1,x_2$. 
  As long as $ z $ is sufficiently regular in space and time so
  that all derivatives exist, we see that $\Delta$ remains zero if it
  is zero initially, as it is the solution of an initial value problem with Lipschitz right-hand side.
 
  \end{proof}

\subsection{A couple of concrete examples}
\label{sec:lozenge}
For the dimer model on the honeycomb lattice with uniform weights, the
characteristic polynomial is \cite{Kenyonnotes} $P(z,w)=z+w-1$ and
(with a conventional choice of coordinates on the honeycomb lattice)
the Newton polygon $N$ is the triangle with vertices
$(0,0),(1,0),(0,1)$. The condition $P(z,w)=0$ implies $w=1-z$, and the
mapping from $z$ to $\nabla h$ induced by \eqref{eq:nablazw} is a
bijection from the upper half complex plane $\mathbb H$ to $N$. The $z$-to-slope mapping is
illustrated in Figure \ref{fig:triangolo}.
\begin{figure}
\begin{center}
\includegraphics[height=4cm]{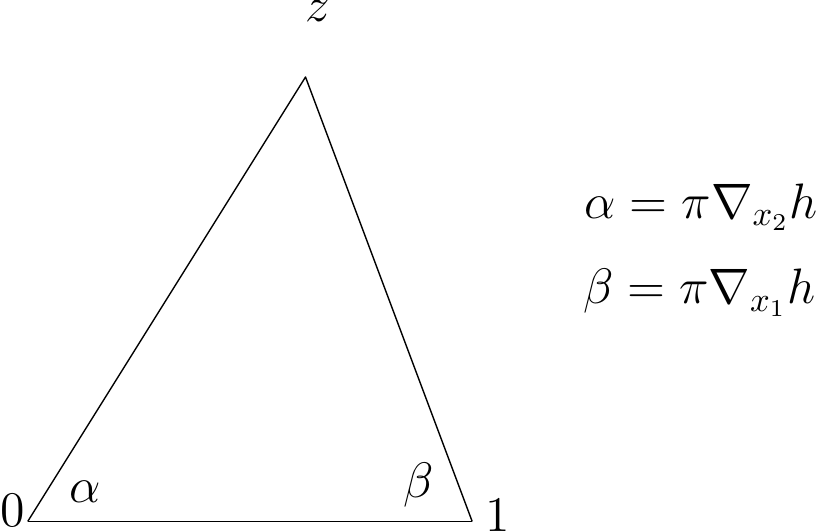}
\caption{The correspondence between $z\in\mathbb H$ and $\nabla h\in N$ for the honeycomb lattice.}
\label{fig:triangolo}
\end{center}
\end{figure}
Some of the two-dimensional growth models mentioned in the
introduction, notably those defined in \cite{BF08} and \cite{BBO}, are
known to preserve the Gibbs measures of the honeycomb dimer model.
For these models, the speed of growth turns out to be a harmonic
function of $z$ that, once expressed in terms of $\nabla h$, can be
checked by direct computation to have AKPZ signature:
$\det (D^2 v)\le0$.  Thanks to our Theorems \ref{thm5} and \ref{thm1},
AKPZ signature follows simply by harmonicity.

Another natural example is when $G$ is the square grid $\mathbb Z^2$
with unit weights. With a natural choice of Kasteleyn matrix one
finds the characteristic polynomial
\[P(z,w)=-1+\frac1z+\frac1w+\frac1{zw}.
\]
Then, $P(z,w)=0$ gives $w=(z+1)/(z-1)$, and one easily checks that relations \eqref{eq:nablazw} give again a
bijection between the upper half-plane and the Newton polygon, that in
this case is the square $N$ with vertices
$(0,0),(1,0),(0,1),$ $(1,1)$. See Fig. \ref{fig:quadrato} and also the discussion in \cite[Sec. 2.4 and Fig. 3]{CFT}.  As we
mentioned in the introduction, the two-dimensional domino-tiling
growth model defined in Section 3.1 of \cite{T15} admits the Gibbs
measures $\pi_\rho$ of the dimer model on $\mathbb Z^2$ as stationary
states. The speed of growth $v(\cdot)$ was computed in \cite{CFT}:
while it has quite a complicated expression in terms of $\nabla h$,
see \cite[Eq. (2.6)]{CFT}, once expressed in terms of $z$ it equals
just $\pi^{-1}\Im z$. In \cite[App. B]{CFT}, a lengthy but direct
computation shows that $\det(D^2 v)\le 0$; again, this can be obtained
as an immediate consequence of our present results.
\begin{figure}
\begin{center}
\includegraphics[height=4cm]{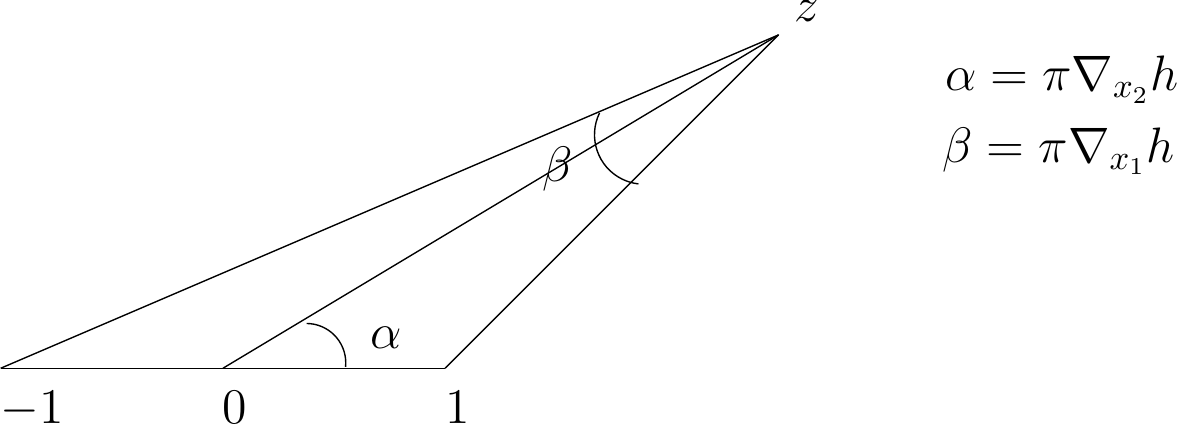}
\caption{The correspondence between $z\in\mathbb H$ and $\nabla h\in N$ for the square grid.}
\label{fig:quadrato}
\end{center}
\end{figure}

\section*{Acknowledgements}
We are grateful to Rick Kenyon, Senya Shlosman and Herbert Spohn for
very helpful suggestions.  A.B. was partially supported by the NSF
grant DMS-1607901 and DMS-1664619.  F.T. was partially supported by
the CNRS PICS grant ``Interfaces al\'eatoires discr\`etes et
dynamiques de Glauber'', by ANR-15-CE40-0020-03 Grant LSD and by
MIT-France Seed Fund ``Two-dimensional Interface Growth and
Anisotropic KPZ Equation''.

\end{document}